\documentclass[notitlepage, longbibliography, superscriptaddress,twocolumn]{revtex4-1}
\usepackage[utf8]{inputenc}
\usepackage{amsmath, amsthm, amsfonts, amssymb, bm, bbm}
\usepackage{hyperref, cleveref}
\usepackage{xcolor}
\usepackage[english]{babel}
\usepackage{comment}

\newtheorem{theorem}{Theorem}

\newtheorem{corollary}{Corollary}

\newtheorem{lemma}{Lemma}

\DeclareMathOperator{\ad}{Ad}

\DeclareMathOperator{\col}{col}
\DeclareMathOperator{\tr}{Tr}
\DeclareMathOperator{\rnk}{rank}

\newcommand{\bb}{\mathbb}
\newcommand{\mc}{\mathcal}

\newcommand{\ket}[1]{| #1 \rangle}

\makeatletter
\def\ketbra#1{\def\tempa{#1}\futurelet\next\ketbra@i}
\def\ketbra@i{\ifx\next\bgroup\expandafter\ketbra@ii\else\expandafter\ketbra@end\fi}
\def\ketbra@ii#1{\left| \tempa \middle\rangle\!\middle\langle #1 \right|}
\def\ketbra@end{\left| \tempa \middle\rangle\!\middle\langle \tempa \right|}
\makeatother

\makeatletter
\def\braket#1{\def\tempa{#1}\futurelet\next\braket@i}
\def\braket@i{\ifx\next\bgroup\expandafter\braket@ii\else\expandafter\braket@end\fi}
\def\braket@ii#1{\left\langle \tempa \middle| #1 \right\rangle}
\def\braket@end{\left\langle \tempa \middle| \tempa \right\rangle}
\makeatother

\newcommand{\dbra}[1]{\ensuremath{\left\langle\!\langle #1\right|}}
\newcommand{\dket}[1]{\ensuremath{\left|#1\rangle\!\right\rangle}}
\newcommand{\cket}[1]{\ensuremath{\left|#1\right)}}
\newcommand{\cketbra}[1]{\ensuremath{\left|#1\right)\!\left(#1\right|}}

\makeatletter
\def\dketbra#1{\def\tempa{#1}\futurelet\next\dketbra@i}
\def\dketbra@i{\ifx\next\bgroup\expandafter\dketbra@ii\else\expandafter\dketbra@end\fi}
\def\dketbra@ii#1{| \tempa \rangle\!\rangle\!\langle\!\langle #1 |}
\def\dketbra@end{| \tempa \rangle\!\rangle\!\langle\!\langle \tempa |}
\makeatother

\makeatletter
\def\dbraket#1{\def\tempa{#1}\futurelet\next\dbraket@i}
\def\dbraket@i{\ifx\next\bgroup\expandafter\dbraket@ii\else\expandafter\dbraket@end\fi}
\def\dbraket@ii#1{\langle\!\langle \tempa | #1 \rangle\!\rangle}
\def\dbraket@end{\langle\!\langle \tempa | \tempa \rangle\!\rangle}
\makeatother

\makeatletter
\def\cketbra#1{\def\tempa{#1}\futurelet\next\cketbra@i}
\def\cketbra@i{\ifx\next\bgroup\expandafter\cketbra@ii\else\expandafter\cketbra@end\fi}
\def\cketbra@ii#1{\left| \tempa \middle)\!( #1 \right|}
\def\cketbra@end{\left| \tempa \middle)\!\middle( \tempa \right|}
\makeatother

\makeatletter
\def\cbraket#1{\def\tempa{#1}\futurelet\next\cbraket@i}
\def\cbraket@i{\ifx\next\bgroup\expandafter\cbraket@ii\else\expandafter\cbraket@end\fi}
\def\cbraket@ii#1{\left( \tempa \middle| #1 \right)}
\def\cbraket@end{\left( \tempa \middle| \tempa \right)}
\makeatother

\begin{document}

\title{Stochastic errors in quantum instruments}
\author{Darian McLaren}
\affiliation{Institute for Quantum Computing, University of Waterloo, Waterloo, Ontario N2L 3G1, Canada}
\affiliation{Department of Applied Mathematics, University of Waterloo, Waterloo, Ontario N2L 3G1, Canada}
\author{Matthew A. Graydon}
\affiliation{Institute for Quantum Computing, University of Waterloo, Waterloo, Ontario N2L 3G1, Canada}
\affiliation{Department of Applied Mathematics, University of Waterloo, Waterloo, Ontario N2L 3G1, Canada}
\author{Joel J. Wallman}
\affiliation{Institute for Quantum Computing, University of Waterloo, Waterloo, Ontario N2L 3G1, Canada}
\affiliation{Department of Applied Mathematics, University of Waterloo, Waterloo, Ontario N2L 3G1, Canada}
\affiliation{Keysight Technologies Canada, Kanata, ON K2K 2W5, Canada}

\begin{abstract}
Fault-tolerant quantum computation requires non-destructive quantum measurements with classical feed-forward.
Many experimental groups are actively working towards implementing such capabilities and so they need to be accurately evaluated.
As with unitary channels, an arbitrary imperfect implementation of a quantum instrument is difficult to analyze.
In this paper, we define a class of quantum instruments that correspond to stochastic errors and thus are amenable to standard analysis methods. 
We derive efficiently computable upper- and lower-bounds on the diamond distance between two quantum instruments.
Furthermore, we show that, for the special case of uniform stochastic instruments, the diamond distance and the natural generalization of the process infidelity to quantum instruments coincide and are equal to a well-defined probability of an error occurring during the measurement.
\end{abstract}

\date{\today}

\maketitle
\section{Introduction}
Quantum technologies are presently prone to noise \cite{preskill2018quantum}; consequently, quantum computations are ridden with errors. Noise thus limits the integrity of quantum computers, since quantum protocols drafted for ideal conditions fail. Imperfect implementations of ideal quantum processes are, however, far from fatal. Fault-tolerant architectures  \cite{FTShor,FTPreskill,aharonov1997fault,knill1998resilient,kitaev2003fault} could perfect quantum information processing \cite{MI,watrous2018theory}. 
Quantum error correction and noise mitigation protocols \cite{gottesman1997stabilizer,temme2017error,endo2018practical,mcardle2019error,maciejewski2020mitigation,koczor2021exponential}, in general, are subject to practical imperfections. 

It is therefore crucial to characterize the quality of the components of experimental fault-tolerant architectures: quantum gates and quantum measurements. The latter, conspicuously, have been studied far less than the former in the realms of quantum certification and validation \cite{eisert2020quantum}. Indeed, there exist a host of randomized benchmarking protocols (with \cite{emerson2005scalable,knill2008randomized,dankert2009exact,magesan2011scalable,magesan2012efficient,francca2018approximate,erhard2019characterizing,helsen2019new,RB,proctor2019direct,merkel2021randomized} a representative list of references) to efficiently characterize the quality of quantum gates. Quantum measurements await a similar suite of efficient benchmarking protocols to validate their applications in large-scale fault-tolerant quantum computing schemes. A significant obstacle has been to pinpoint canonical figures of merit for implementations of quantum measurements, such as the non-destructive quantum measurements with classical feed-forward applied in surface codes. Indeed, such figures of merit are necessary prerequisites for any quality characterization protocols applicable to fault-tolerant quantum computers.

In this paper, we establish that the current canonical figures of merit for implementations of quantum gates (namely, their process infidelities and diamond distances) capture the quality of quantum measurements as well. An arbitrary implementation of a quantum measurement can be tailored to a uniform stochastic implementation \cite{beale2023randomized}, which is similar to the fact that generic gate noise can be tailored to a stochastic channel \cite{Graydon2022}. We thus set our analysis within a framework for quantum instruments \cite{davies1970operational,heinosaari2011mathematical,wilde2013quantum} and we focus on uniform stochastic instruments.

We prove that a uniform stochastic implementation of a quantum instrument is such that its diamond distance to the ideal L{\"u}ders instrument equals its process infidelity (\cref{thm:stochasticDiamondDistance}.) In the general case, we consider arbitrary implementations and derive efficiently computable lower and upper bounds on the diamond distance (\cref{thm:GeneralInstrumentDistance}.) We show that uniform stochastic instruments saturate the lower bound; moreover, their diamond distance is just the probability that an error occurs during measurement. The formal quality of a uniform stochastic instrument thus admits a surprisingly simple operational description. Our proofs hinge on several complementary observations. In particular, we generalize a main theorem from \cite{Graydon2022} to the case of unnormalized stochastic channels (\cref{thm:unnormalizedStochastic}) and we derive consequences of a main theorem from \cite{mclaren2022} (\cref{cor:nonuniform} and \cref{cor:uniform}) to prove our main results. We also establish an improved generalized lower bound for the famous Fuchs-van de Graff inequalities (\cref{thm:Fuchs})

Our results herein set the mathematical backbone for efficient protocols to characterize quantum instruments. We pinpoint the process infidelity and diamond distance as canonical figures of merit for noisy implementations of quantum measurements. In a sequel \cite{Mahmoud2023}, we expound a randomized benchmarking protocol to efficiently extract the process infidelity. We remark that quantum process tomography requires resources that scale exponentially with system size.  We note that R.\ Stricker \textit{et al}.~\cite{stricker2022characterizing} recently proposed a method to characterize the quality of noisy quantum instruments via process tomography.

The balance of this paper is organized as follows. In \cref{sec:preliminaries}, we set our notation and review the uniform stochastic instruments defined in \cite{beale2023randomized}. In \cref{sec:diamonds}, we consider unnormalized (\textit{i.e}.\ generally trace decreasing) stochastic channels and generalize a main theorem in \cite{Graydon2022}.
In \cref{sec:InstrumentFidelity}, we show that the process fidelity of a uniform stochastic instrument to its ideal implementation is simply the probability of implementing the ideal instrument.
In \cref{sec:instDiamonds}, we establish upper and lower bounds on the diamond distance between any implementation of a quantum measurement and its ideal L{\"u}ders form. Next, in \cref{sec:stochDiamonds}, we derive a closed form exact expression for the diamond distance of a stochastic implementation of a quantum measurement and its ideal L{\"u}ders form. We conclude this paper in \cref{sec:conclusion} and point out avenues for future research. In \cref{sec:Fuchs}, we prove \cref{thm:Fuchs}.

\section{Preliminaries}
\label{sec:preliminaries}
We will follow the notation of \cite{beale2023randomized} closely and so briefly review the relevant notation.

\subsection{Quantum states}

Let $\bb{H}_{D,E}$ be the Hilbert space of bounded linear operators from $\bb{C}^D$ to $\bb{C}^E$, where we will use the shorthand $\bb{H}_{D} = \bb{H}_{D, D}$.
Further, let $\bb{D}_D \subset \bb{H}_D$ be the space of density operators.
A quantum state can then be represented either as a density matrix, $\rho \in \bb{D}_D$, or, in vectorized form as $\col{\rho} \in \bb{C}^{D^2}$.
The particular choice of vectorization map is unimportant for a fixed state space provided it is consistent and satisfies $\dbraket{A}{B} = \tr A^\dagger B$ where $\dbra{A} = \dket{A}^\dagger$.
However, as we work with different state spaces when using the Choi-Jamio{\l}kowski isomorphism~\cite{CHOI1975285,JAMIOLKOWSKI1972275}, we choose $\dket{A} = \col(A)$ where $\col$ denotes the column stacking map.
The column stacking map satisfies the vectorization identities
\begin{align}\label{eq:VectorizationIdentities}
    \col(ABC) &=(C^T\otimes A)\col(B) \notag\\
    \tr(A^\dag B) &=\col(A)\col^\dagger(B).
\end{align}

As we are primarily concerned with $n$-qudit computational basis measurements, we define $D = d^n$ and
\begin{align}
    \cket{} : \bb{Z}_D \to \bb{C}^{D^2} :: \cket{j} = \col(\ketbra{j}) = \ket{jj}.
\end{align}
We also make heavy use of the maximally entangled state
\begin{align}
    \Phi = \frac{1}{D} \sum_{j, k \in \bb{Z}_D} \ketbra{jj}{kk},
\end{align}
which, by linearity, can also be written as
\begin{align}\label{eq:PhiState}
    \Phi = \frac{1}{D} \col(I_D) \col^\dagger(I_D),
\end{align}
where $I_D \in \bb{H}_D$ is the identity map.
Note that we choose the column-stacking vectorization map as \cref{eq:PhiState} only holds for the column- and row-stacking vectorization maps.

\subsection{Quantum operations}

Quantum operations are linear maps from quantum states to quantum states.
We can write any linear map $\mc{C} : \bb{H}_D \to 
 \bb{H}_E$ as
\begin{align}\label{eq:GeneralKrausDecomposition}
    \mc{C}(\rho) = \sum_j L_j \rho R_j^\dagger
\end{align}
where $L_j, R_j \in \bb{H}_{D, E}$.
In superoperator form, we have
\begin{align}
    \mc{C} = \sum_j \ad_{L_j, R_j}.
\end{align}
We are primarily interested in completely positive maps, for which we can choose $R_j = L_j$ for all $j$~\cite{CHOI1975285} and so we use the short-hand $\ad_L = \ad_{L,L}$.
We can also represent a linear map using the Choi state
\begin{align}
    \mc{J}(\mc{C}) = (\ad_{I_d} \otimes \mc{C})(\Phi),
\end{align}
which, using \cref{eq:PhiState,eq:VectorizationIdentities}, can be re-written as
\begin{align}\label{eq:ChoiState}
    \mc{J}(\mc{C}) = \frac{1}{D} \sum_j \col(L_j) \col^\dagger(R_j).
\end{align}
As the decomposition of a linear map as in \cref{eq:GeneralKrausDecomposition} is not unique, we define the Kraus rank $\kappa$ to be minimum number of Kraus operators necessary in the decomposition.
One can prove that
\begin{align}
    \kappa(\mc{C}) = \rnk J(\mc{C})
\end{align}
where $\rnk(M)$ denotes the matrix rank of a matrix $M$, that is, the number of nonzero singular values of $M$.

\subsection{Quantum instruments}

A quantum instrument is a special type of quantum channel that may be used to represent a measurement process on a quantum system. 
The action of a quantum instrument on a state $\rho$ can be written as
\begin{align}\label{eq:GeneralInstrument}
    \mc{M}(\rho) = \sum_j \pi_j \rho \pi_j \otimes \ketbra{j}
\end{align}
for some set of orthogonal projectors $\{\pi_j\}$ that sum to the identity, where $\ketbra{j}$ represents the observed outcome.
That is, $\mc{M}$ is a quantum channel with Kraus operators $\pi_j \otimes \ket{j}$.
In the superoperator representation, we have
\begin{align}
    \mc{M} = \sum_j \ad_{\pi_j} \otimes \cket{j}.
\end{align}
An imperfect implementation of an instrument $\mc{M}$ can be written as
\begin{align}\label{eq:GeneralSubsystemMeasurementError}
    \Theta(\mc{M}) = \sum_j \mc{M}_j \otimes \cket{j},
\end{align}
where the $\mc{M}_j$ are completely positive maps.
The probability $p(j|\Theta(\mc{M}))$ of obtaining an outcome $j$ when a system in the state $\rho$ is measured using the instrument $\Theta(\mc{M})$ is determined by Born's rule,
\begin{align}\label{eq:BornRule}
p(j|\Theta(\mc{M})) = \tr \mc{M}_j(\rho).
\end{align}
We are particularly interested in subsystem measurements, wherein only a subset of qudits are measured with a rank 1 measurement~\cite{beale2023randomized}.
Without loss of generality we assume that the unmeasured qudits are idle and so take $\pi_j = I_E \otimes \ketbra{j}$ where $\bb{H}_E$ is the state space of the unmeasured qudits.
That is, we define a subsystem measurement to be a quantum instrument of the form
\begin{align}\label{eq:SubsystemMeasurement}
\mc{M} = \sum_{j \in \bb{Z}_D} \ad_{I_E} \otimes \cketbra{j} \otimes \cket{j}.
\end{align}
A general error model for a subsystem measurement is still in the form of \cref{eq:GeneralSubsystemMeasurementError}, that is, a general error model may entangle the state of the measured and unmeasured qudits.
However, as with unitary channels, we would like to restrict to stochastic errors in subsystem measurements because they are easier to analyze.
There are three types of stochastic errors that could occur during an implementation of a subsystem measurement:
\begin{enumerate}
    \item reporting the wrong outcome;
    \item flipping the state of the measured qudit after the measurement; and
    \item applying a stochastic error to the unmeasured qudits~\cite{Graydon2022}.
\end{enumerate}
In general, the errors applied could be correlated with the observed measurement outcome, which substantially complicates the analysis.
Fortunately, these correlations can be removed by a recently introduced generalization of randomized compiling~\cite{beale2023randomized}, which tailors a generic implementation described by \cref{eq:GeneralSubsystemMeasurementError} to a uniform stochastic implementation described by the simplified channel
\begin{align}\label{eq:UniformStochasticImplementation}
    \Theta(\mc{M}) = \sum_{a, b, j \in \bb{Z}_D} \mc{T}_{a, b} \otimes \cketbra{j + a}{j + b} \otimes \cket{j},
\end{align}
where each $\mc{T}_{a, b}$ is an unnormalized (i.e., generally trace-decreasing) stochastic channel.
One of the key features of \cref{eq:UniformStochasticImplementation} is that the error applied to the unmeasured qubits is indepenent of $j$, although it may depend upon the error in the measurement itself, that is, on the values of $a$ and $b$.
We will show that uniform stochastic channels generalize stochastic channels in that there is a closed form expression for the diamond distance that does not involve any maximization or minimization.
For this, it will be convenient to write
\begin{align}\label{eq:Normalization}
    \mc{T}_{a, b} = \nu_{a, b} \mc{T}'_{a, b}
\end{align}
where $\nu : \bb{Z}_D^2 \to [0, 1]$ is a probability distribution and $\mc{T}'_{a, b}$ is a trace-preserving map, that is, a normalized stochastic channel. If we lift the restriction of uniformity, then the stochastic errors applied to the unmeasured qudits depend on the measurement outcome and we have
\begin{align}\label{eq:StochasticImplementation}
    \Theta(\mc{M})=\sum_{a, b, j \in \bb{Z}_D} \mc{T}_{a, b,j} \otimes \cketbra{j + a}{j + b} \otimes \cket{j}\text{,}
\end{align}
where, as above, we can write each $\mc{T}_{a,b,j}=\nu_{a,b,j}\mc{T}'_{a,b,j}$ such that $\nu:\mathbb{Z}_{D}^{3}\to[0,1]$ a probability distribution and each $\mc{T}'_{a,b,j}$ is a normalized stochastic channel.

\section{Diamond distances for unnormalized stochastic channels}
\label{sec:diamonds}
We now obtain an expression for the diamond distance between an unnormalized stochastic channel, and the identity channel, generalizing the results of \cite{Graydon2022}.

Recall that a square root of an operator $A$ is an operator $B$ such that $A = B^2$ and that if $A$ is positive semi-definite, then it has a unique positive semi-definite square root, referred to as the square root.
Recall that the trace norm of an operator $A$ is defined to be
\begin{align}\label{eq:TraceNorm}
    \| A \|_1 = \tr \sqrt{A A^\dagger},
\end{align}
and that the trace distance between two operators $A$ and $B$ is the trace norm of their difference.
We will frequently use the following two properties of the trace norm.
For any two operators $A$ and $B$,
\begin{align}\label{eq:NormTensorProduct}
    \| A \otimes B \|_1 = \| A \|_1 \| B \|_1,
\end{align}
and for any positive operator $A$, we have
\begin{align}\label{eq:NormPostiveOperator}
    \| A \|_1 = \tr A.    
\end{align}
The diamond norm of a Hermiticity-preserving map is
\begin{align}\label{eq:diamond}
    \|\mc{A}\|_\diamond= \max_{\rho \in \bb{H}_D} \| \mc{I}_d \otimes \mc{A}(\rho)\|_1,
\end{align}
where the restriction to density operators is only valid for Hermiticity-preserving maps~\cite{watrous2018theory}.
The diamond distance is then the natural distance induced by the diamond norm. 
As we focus on comparing a channel to the identity in this section, we define the diamond distance of an arbitrary channel $\mc{A}$ from the identity to be
\begin{align}
    r_\diamond(\mc{A}) = \frac{1}{2}\|\mc{A}-\mc{I}_d\|_\diamond.
\end{align}

Before proceeding, we need the following lemma.

\begin{lemma}\label{lem:orthogonality}
    For any set of orthogonal Hermitian operators $\{M_j\}$, we have
    \begin{align*}
        \|\sum_j M_j \|_1 = \sum_j \| M_j \|_1.
    \end{align*}
\end{lemma}

\begin{proof}
Let $M = \sum_j M_j$ and note that $M M^\dagger = \sum_j M_j M_j^\dagger$ by assumption.
Now note that the $M_j M_j^\dagger$ are positive semidefinite operators and that the square roots of orthogonal positive semidefinite operators are also orthogonal.
Therefore we have
\begin{align*}
    \left(\sum_j \sqrt{M_j M_j^\dagger}\right)^2 
    &= \sum_{j, k} \sqrt{M_j M_j^\dagger} \sqrt{M_k M_k^\dagger} \notag\\
    &= \sum_j M_j M_j^\dagger \notag\\
    &= M M^\dagger.
\end{align*}
Taking the square roots of both sides gives
\begin{align}
    \sqrt{M M^\dagger} = \sum_j \sqrt{M_j M_j^\dagger}
\end{align} 
and so the claim follows by the linearity of the trace.
\end{proof}

We now obtain the following expression for the diamond distance from the identity for an unnormalized stochastic channel.
As is conventional, we state the lower bound in terms of the process fidelity of the channel to the identity, where the process fidelity between two channels $\mc{A}$ and $\mc{B}$ is~\cite{Schumacher1996}
\begin{align}\label{eq:Fidelity}
    \mc{F}(J(\mc{A}), J(\mc{B})) = \| \sqrt{J(\mc{A})} \sqrt{J(\mc{B})} \|_1^2.
\end{align}

\begin{theorem}\label{thm:unnormalizedStochastic}
Let $\mc{T}$ be an unnormalized stochastic channel. Then 
\begin{align}
    r_{\diamond}(\mc{T})=\frac{1 + \|J(\mc{T})\|_1}{2} - \mc{F}(J(\mc{T}),J(\mc{I}_d)).
\end{align}
\end{theorem}
\begin{proof}
    To begin, we let $\mc{T}=\nu\mc{T}'$ for a normalized stochastic channel $\mc{T}$. Note that as $\mc{T}'$ is a quantum channel it admits a Kraus decomposition with Kraus operators $\{B_k\}_{k\in K}$, where as $\mc{T}'$ is a stochastic channel, we may assume that $B_0=\sqrt{\lambda} I_d$ from some $\lambda\in[0,1]$. 
    We then proceed by first showing that
    \begin{align}\label{eq:Up}
        \|\mc{T}-\mc{I}_d\|_\diamond\leq 1 + \nu - 2\nu\lambda.
    \end{align}

    Using the triangle inequality combined with the homogeneity of the diamond norm, and $\|I_d\|_\diamond=1$ we obtain
    \begin{align}\label{eq:UpperBound}
        \|\mc{T}-\mc{I}_d\|_\diamond&=\|(\nu\lambda - 1)I_d - \nu\sum_{k\neq 0} \ad_{B_k}\|_\diamond \notag\\
        &\leq1-\nu\lambda + \nu\|\sum_{k\neq 0} \ad_{B_k}\|_\diamond.
    \end{align}
    Applying the definition of the diamond norm, and the triangle inequality we have
    \begin{align}
        \|\sum_{k\neq 0} \ad_{B_k}\|_\diamond&\leq\sup_{\rho}\sum_{k\neq 0}\| \ad_{B_k\otimes I_d}(\rho)\|_1\notag\\
        &=\sup_{\rho}\sum_{k\neq 0} \tr\ad_{B_k\otimes I_d}(\rho)\notag\\
        &=\sup_{\rho}1- \tr\ad_{B_0\otimes I_d}(\rho)\notag\\
        &=1-\lambda,
    \end{align}
    which when substituted into \cref{eq:UpperBound} gives the desired result.
    
    We now show that
    \begin{align}\label{eq:Low}
        \|\mc{T}-\mc{I}_d\|_\diamond\geq 1 + \nu - 2\nu\lambda.
    \end{align}
To begin we use lemma 7 of \cite{RBCon} to lower bound the diamond distance in terms of the Choi matrices and then apply \cref{eq:ChoiState}
\begin{align}
    \|\mc{T} - \mc{I}_d\|_\diamond&\geq\|J(\mc{T}) - J(\mc{I}_d)\|_1\notag\\
    &=\frac{1}{d}\Big\|(1-\nu\lambda)\col(I_d)\col(I_d)^\dag \notag\\
    &\quad\quad  + \nu\sum_{k\neq 0}\col(B_k)\col(B_k)^\dag\Big\|_1 .
\end{align}
Recall that from $\mc{T}$ being a stochastic channel the Kraus operators of $\mc{T}$ satisfy the following orthogonality condition
\begin{align}
    \tr(B_k^\dag B_l)=\delta_{kl}\tr(B_k^\dag B_k).
\end{align}
Consequently, the operators $\{\col(B_k) \col^\dagger(B_k)\}_{k\in K}$ are mutually orthogonal hermitian operators. 
Therefore applying \cref{lem:orthogonality} we have

\begin{align}
    \|\mc{T} - \mc{I}_d\|_\diamond&\geq1-\nu\lambda + \frac{\nu}{d}\sum_{k\neq 0}\col(B_k)^\dagger \col(B_k)\notag\\
    &=1-\nu\lambda + \frac{\nu}{d}\sum_{k\neq 0}\col(B_k)^\dagger \col(B_k)\notag\\
    &=1-\nu\lambda + \frac{\nu}{d}(d - \col(B_0)^\dag\col(B_0)),
\end{align}
and the result follows from the assumption $B_0=\sqrt{\lambda}I_d$

From equations \ref{eq:Up} and \ref{eq:Low} we have
\begin{align}
    \|\mc{T}-\mc{I}_d\|_\diamond= 1 + \nu - 2\nu\lambda.
\end{align}
Hence we complete the proof by noting that
\begin{align}
    \mc{F}(J(\mc{T}),J(\mc{I}_d)) = \nu\lambda,
\end{align}
and
\begin{align}
    \|J(\mc{T})\|_1=\nu.
\end{align}
\end{proof}

\section{Process infidelity of quantum instruments}\label{sec:InstrumentFidelity}

The diamond distance between a stochastic channel and the identity is given by the process fidelity of a stochastic channel to the identity~\cite{Graydon2022}.
We will prove that this statement generalizes to uniform stochastic implementations of subsystem measurements, and also demonstrate that it does not generalize to arbitrary stochastic implementations of subsystem measurements.
Towards this goal, we now compute the process fidelity for stochastic implementations of subsystem measurements using the following lemmas~\cite{mclaren2022}.

We remark that $\nu_{0, 0}$ is the probability that the effect of the measurement on the measured system is perfect, that is, that the correct state is detected and that the system is left in the correct state.
Furthermore, $\lambda_{0, 0}$ is the probability that no error affects the unmeasured system conditioned on the ideal measurement being performed on the measured system.
Therefore \cref{cor:uniform} establishes that the fidelity between $J(\mc{M})$ and $J(\Theta(\mc{M}))$ is simply the joint probability of no errors occurring.

\begin{lemma}\label{lem:fidelSums}
    Let $\mc{A}=\sum_j \mc{A}_j \otimes \cket{j}$ and $\mc{B}=\sum_j \mc{B}_j \otimes \cket{j}$ be quantum instruments then
    \begin{align}
        \sqrt{\mc{F}(J(\mc{M}),J(\Theta(\mc{M})))} = \sum_j \sqrt{\mc{F}(J(\mc{A}_j),J(\mc{B}_j))}
    \end{align}
\end{lemma}

\begin{lemma}\label{cor:nonuniform}
 Let $\mc{M}$ be a subsystem measurement, as in \cref{eq:SubsystemMeasurement}, and let $\Theta(\mc{M})$ be a nonuniform stochastic implementation of $\mc{M}$, as in \cref{eq:StochasticImplementation}. Identify $\lambda_{0,0,j}$ via $J(\mc{T}'_{0,0,j})\mathrm{col}(I_{E})=\mathrm{col}(I_{E})\lambda_{0,0,j}$. Then
    \begin{align}\label{eq:nonUniformFidel}
        \sqrt{\mc{F}(J(\mc{M}),J(\Theta(\mc{M})))}=\frac{1}{D}\sum_{j\in\mathbb{Z}_{D}}\sqrt{\nu_{0,0,j}\lambda_{0,0,j}}\text{.}
    \end{align}
\end{lemma} 
\begin{proof}
For all $j \in \bb{Z}_D$, we define $\pi_j = I_E \otimes \ketbra{j}$ the corresponding completely positive map
\begin{align}
    \mc{A}_j : \bb{H}_{DE} \to \bb{H}_{DE} :: \mc{A}_j(\rho) = \pi_j \rho \pi_j.
\end{align}
By \cref{lem:fidelSums}, we have
\begin{align}
     \sqrt{\mc{F}(J(\mc{M}),J(\Theta(\mc{M})))} = \sum_{j \in \bb{Z}_D} \sqrt{\mc{F}(J(\mc{A}_j),J(\mc{M}_j))}.
\end{align}
From \cref{eq:ChoiState}, we have
\begin{align}
    \sqrt{J(\mc{A}_j)} = \sqrt{\frac{\tr \pi_j}{DE}} v_j v_j^\dagger
\end{align}
where $v_j = \col(\pi_j) / \sqrt{\tr \pi_j}$ is a unit vector.
By \cref{eq:Fidelity} and the homogeneity of the trace norm, for each $j \in \bb{Z}_D$ we then have
\begin{align}
    \sqrt{\mc{F}(J(\mc{A}_j),J(\mc{M}_j))} 
    &= \|\sqrt{J(\mc{A}_j)} \sqrt{J(\mc{M}_j)} \|_1 \notag\\
    &= \sqrt{\frac{\tr \pi_j}{DE}} \| v_j v_j^\dagger \sqrt{J(\mc{M}_j)} \|_1 \notag\\
    &= \sqrt{\frac{\tr \pi_j}{DE}} \sqrt{v_j^\dagger J(\mc{M}_j) v_j}.
\end{align}
To simplify this further, let $\{B^{(\alpha)}_{a,b,j}\}_{\alpha=1}^{E^{2}}$ be Hilbert-Schmidt orthogonal Kraus operators for $\mc{T}_{a,b,j}$ and $c^{(\alpha)}_{a,b,j} = \col(B^{(\alpha)}_{a,b,j} \otimes \ketbra{j+a}{j+b})$.
Then from \cref{eq:ChoiState}, we have
\begin{align}
    J(\mc{M}_j) = \frac{1}{DE} \sum_{a, b \in \bb{Z}_D} c^{(\alpha)}_{a, b, j} (c^{(\alpha)}_{a, b, j})^\dagger
\end{align}
and that
\begin{align}
        v_j^\dagger c^{(\alpha)}_{a, b, j} = \delta_{a,0}\delta_{b,0} \tr B^{(\alpha)}_{a,b,j} / \sqrt{\tr \pi_j}. 
\end{align}
Up to now, we have not needed the assumption that $\mc{T}_{j, a, b}$ is an unnormalized stochastic channel.
Adding in this assumption, we have
\begin{align}
    B^{(1)}_{a,b,j}=I_{E}\sqrt{\nu_{a, b, j}\lambda_{a,b,j}},
\end{align}
and all other Kraus operators are traceless, which gives
\begin{align}
    \sqrt{\mc{F}(J(\mc{A}_j),J(\mc{M}_j))} 
    &= \frac{1}{D} \sum_{j \in \bb{Z}_D} \frac{1}{E}\tr B^{(1)}_{0, 0, j} \notag\\
    &= \frac{1}{D} \sum_{j \in \bb{Z}_D} \sqrt{\nu_{a, b, j}\lambda_{a,b,j}}
\end{align}
as claimed.
\end{proof}
\begin{corollary}\label{cor:uniform}
 Let $\mc{M}$ be a subsystem measurement, as in \cref{eq:SubsystemMeasurement}, and let $\Theta(\mc{M})$ be a uniform stochastic implementation of $\mc{M}$, as in \cref{eq:UniformStochasticImplementation}. Identify $\lambda_{0,0}$ via $J(\mc{T}'_{0,0})\mathrm{col}(I_{E})=\mathrm{col}(I_{E})\lambda_{0,0}$. Then
 \begin{align}\label{eq:UniformFidel}
        \mc{F}(J(\mc{M}),J(\Theta(\mc{M})))=\nu_{0,0}\lambda_{0,0}\text{.}
    \end{align}
\end{corollary} 
\begin{proof}
    Immediate from \cref{cor:nonuniform} in light of uniformity, that is $\nu_{0,0,j}=\nu_{0,0}$ and $\lambda_{0,0,j}=\lambda_{0,0}$ are independent of the measurement outcome $j\in\mathbb{Z}_{D}$.
\end{proof}

\section{Diamond distance between quantum instruments}
\label{sec:instDiamonds}

We now obtain lower and upper bounds on the diamond distance between the implementation of an instrument and its ideal.
The lower bound is more complicated than for unitary channels, as it requires analyzing multiple measurement outcomes.

\begin{theorem}\label{thm:GeneralInstrumentDistance}
    Let $\mc{M}$ be a subsystem measurement as in \cref{eq:SubsystemMeasurement}, $\pi_j = I_E \otimes \ketbra{j}$, $\Theta(\mc{M})$ be an implementation of $\mc{M}$ in the form \cref{eq:GeneralSubsystemMeasurementError}, and $\Delta = \Theta(\mc{M}) - \mc{M}$.
    Then for any $\sigma \in \bb{D}_{E}$ and any $j \in \bb{Z}_D$ with $\sigma_j = \sigma \otimes \ketbra{j}$, the diamond distance between $\mc{M}$ and $\Theta(\mc{M})$ satisfies
    \begin{align*}
        \| \Delta \|_\diamond &\geq 1 - \tr \mc{M}_j(\sigma_j) + \|\mc{M}_j(\sigma_j) - \sigma_j\|_1 
    \end{align*}
    and
    \begin{align*}
        \| \Delta \|_\diamond \leq DE\sum_k\|J(\mc{M}_k) - J(\ad_{\pi_k})\|_1.
    \end{align*}
\end{theorem}

\begin{proof}
    To prove the lower bound, note that $\sigma_j \in \bb{D}_{DE}$ is a density operator for all $j$ and so for all $j$ we have
    \begin{align}
        \| \Delta \|_\diamond 
        &\geq \|\mc{I}_{DE}(\sigma_j) \otimes \Delta(\sigma_j) \|_1 \geq \| \Delta(\sigma_j) \|_1
    \end{align}
    by \cref{eq:NormPostiveOperator,eq:NormTensorProduct}.
    By assumption,
    \begin{align}
        \Delta(\sigma_j) = \sum_k (\mc{M}_k - \ad_{\pi_k})(\sigma_j) \otimes \ketbra{k}.
    \end{align}
    By \cref{lem:orthogonality,eq:NormTensorProduct}, we then have
    \begin{align}
       \|\Delta(\sigma_j) \|_1 
       &= \sum_k \|\mc{M}_k(\sigma_j) - \delta_{j, k}\sigma_j \|_1 .
    \end{align}
    By \cref{eq:NormPostiveOperator} and as $\sum_j \mc{M}_j$ is a trace-preserving map, we then have
    \begin{align}\label{eq:ActionOnSigmaJ}
       \|\Delta(\sigma_j) \|_1 
       &= \| \mc{M}_j(\sigma_j) - \sigma_j \|_1 + 1 - \tr \mc{M}_j(\sigma_j),
    \end{align}
    thus obtaining the lower bound.
    The upper bound follows from the standard upper bound obtained using the Choi state~\cite{Wallman2014} (noting that the dimensional factor is from the input dimension) together with \cref{lem:orthogonality}.
\end{proof}

\section{Stochastic implementations of subsystem measurements}
\label{sec:stochDiamonds}
We now prove that the diamond distance between a uniform stochastic implementation of a subsystem measurement and its corresponding ideal saturates the lower bound in \cref{thm:GeneralInstrumentDistance} and provide an explicit formula and interpretation.
The proof that the lower bound is saturated is essentially trivial and can readily be generalized to include errors that depend on the measurement outcome, however, the explicit formula would then require a maximization over measurement outcomes. In light of \cref{sec:InstrumentFidelity}, the following theorem establishes that the diamond distance of a uniform stochastic implementation $\Theta(\mc{M})$ to $\mc{M}$ is, quite simply, the probability that an error occurs.

\begin{theorem}\label{thm:stochasticDiamondDistance}
    Let $\mc{M}$ be a subsystem measurement as in \cref{eq:SubsystemMeasurement} and $\pi_j = I_E \otimes \ketbra{j}$. Then for any uniform stochastic implementation $\Theta(\mc{M})$ of $\mc{M}$ in the form \cref{eq:UniformStochasticImplementation}, the diamond distance between $\mc{M}$ and $\Theta(\mc{M})$ is
    \begin{align*}
       \frac{1}{2} \| \Theta(\mc{M}) - \mc{M} \|_\diamond &= 1-\mc{F}(J(\mc{M}),J(\Theta(\mc{M}))).
    \end{align*}
\end{theorem}

\begin{proof}
    Let
    \begin{align}\label{eq:EntanglementBreaking}
        \mc{E} = \sum_{j \in \bb{Z}_D} \ad_{\pi_j},
    \end{align}
    which corresponds to performing the ideal measurement and forgetting the outcome.
    Note that
    \begin{align}
        \Delta = \Delta \circ \mc{E}
    \end{align}
    and so for any $\rho \in \bb{D}_{DE}$ we have
    \begin{align}
        \| \Delta(\rho) \|_1 = \| \Delta (\tilde{\rho}) \|_1
    \end{align}
    where $\tilde{\rho} = \mc{E}(\rho)$.
    From \cref{eq:EntanglementBreaking}, $\tilde{\rho} = \sum_j \alpha_j \sigma^{(j)} \otimes \ketbra{j}$ for some probability distribution $\alpha : \bb{Z}_D \to [0, 1]$ and density operators $\sigma^{(j)} \in \bb{D}_E$.
    Therefore, by the triangle inequality, $\| \Delta(\rho) \|_1$ is maximized by a state $\rho = \sigma_j = \sigma \otimes \ketbra{j}$.
    Noting that \cref{eq:ActionOnSigmaJ} is an equality and defining $\mu_j = \mc{M}_j(\sigma_j)$, we thus have
    \begin{align}\label{eq:InducedNorm}
        \max_{\rho \in \bb{D}_{ED}} \|\Delta(\rho)\|_1 = 1 + \max_{\stackrel{\sigma \in \bb{D}_E}{j \in \bb{Z}_D}}  \| \mu_j - \sigma_j \|_1 - \tr \mu_j.
    \end{align}
    However, we have yet to account for the tensor factor in the definition of the diamond distance.

    Before accounting for the additional identity channel in \cref{eq:diamond}, we first use the assumption that $\Theta(\mc{M})$ is a uniform stochastic implementation of $\mc{M}$ to simplify \cref{eq:InducedNorm}.
    Equating terms between \cref{eq:GeneralSubsystemMeasurementError,eq:UniformStochasticImplementation}, we have
\begin{align}
    \mc{M}_j = \sum_{a, b \in \bb{Z}_D} \mc{T}_{a, b} \otimes \cketbra{j + a}{j + b}.
\end{align}
Therefore we have 
\begin{align}
    \mu_j = \sum_{a \in \bb{Z}_D} \mc{T}_{a, 0}(\sigma) \otimes \ketbra{j + a}.
\end{align}
    Using \cref{eq:NormPostiveOperator,eq:NormTensorProduct}, we have
\begin{align}
    \| \mu_j - \sigma_j \|_1 &= \| \mc{T}_{0, 0}(\sigma) - \sigma \|_1 + \sum_{a \in \bb{Z}_D : a > 0} \tr \mc{T}_{a, 0}(\sigma) \notag\\
    &= \| \mc{T}_{0, 0}(\sigma) - \sigma \|_1 + \tr \mu_j - \tr \mc{T}_{0, 0}(\sigma).
\end{align}
Therefore with \cref{eq:Normalization}, \cref{eq:InducedNorm} simplifies to
\begin{align}\label{eq:SimplifiedInducedNorm}
    \max_{\rho \in \bb{D}_{ED}} \|\Delta(\rho)\|_1 &= 1 + \max_{\sigma \in \bb{D}_E}  \| \mc{T}_{0,0}(\sigma) - \sigma \|_1 - \tr \mc{T}_{0, 0}(\sigma) \notag\\
    &= 1 - \nu_{0, 0} + \max_{\sigma \in \bb{D}_E}  \| \mc{T}_{0,0}(\sigma) - \sigma \|_1.
\end{align}
To include the identity channel, we can simply use the above arguments with $\mc{M} \to \mc{I}_f \otimes\mc{M}$ and $\Theta(\mc{M}) \to \mc{I}_f \otimes \Theta(\mc{M})$ for any dimension $f$ and note that this mapping does not change the value of $\nu_{0, 0}$.
Therefore
\begin{align}\label{eq:SimplifiedDiamondNorm}
    \|\Delta\|_\diamond &= 1 - \nu_{0, 0} + \| \mc{T}_{0,0} - \mc{I}_E \|_\diamond.
\end{align}
In light of \cref{thm:unnormalizedStochastic} and Eq.(21) in \cite{Graydon2022} we have that
\begin{align}
    \| \mc{T}_{0,0} - \mc{I}_E \|_\diamond=1+\nu_{0,0}-2\nu_{0,0}\lambda,
\end{align}
whence 
\begin{align}
    \frac{1}{2}\|\Delta\|_{\diamond}=1-\nu_{0,0}\lambda.
\end{align}
Substituting in the process fidelity via \cref{cor:uniform} completes the proof.
\end{proof}

\section{On the necessity of uniformity}

We now show that the assumption that the implementation is a uniform stochastic channel is required to obtain \cref{thm:stochasticDiamondDistance}.
We will show this by constructing a non-uniform stochastic channel where the diamond distance is not given by \cref{thm:stochasticDiamondDistance}.
Consider a non-uniform stochastic channel where the only error is a stochastic error that affects the unmeasured qudits but the error depends upon the observed outcome, that is,
\begin{align}
    \Theta(\mc{M}) = \sum_{j\in\mathbb{Z}_d} \mc{T}_j \otimes \cketbra{j} \otimes \cket{j},
\end{align}
where each $\mc{T}_j$ is a stochastic channel.

We now compute $\|\Theta(\mc{M}) - \mc{M}\|_\diamond$.
We first note that \cref{eq:InducedNorm} is valid for any implementation $\Theta(\mc{M})$, and hence for any positive integer $F$ we have
    \begin{align}
        \max_{\rho \in \bb{D}_{DEF}} \|\Delta_F(\rho)\|_1 = 1 + \max_{\stackrel{\sigma \in \bb{D}_{FE}}{j \in \bb{Z}_D}}  \| \mu_j - \sigma_j \|_1 - \tr \mu_j,
    \end{align}
where
\begin{align}
    \Delta_F &= \mc{I}_F\otimes (\Theta(\mc{M})-\mc{M})\notag\\
    \mc{M}_j&=\mc{I}_F\otimes\mc{T}_j\otimes \ad_{\ketbra{k}}\notag\\
    \sigma_j &= \sigma \otimes\ketbra{j},\quad \sigma\in\mathbb{D}_{DF}\notag\\
    \mu_j &= \mc{M}_j(\sigma_j).
\end{align}
Since the stochastic channels $\mc{T}_j$ are trace preserving we have $\tr\mu_j=1$, and hence with the above definitions we simplify to
\begin{align}\label{eq:NonUniformDiamond1}
    \max_{\rho \in \bb{D}_{DEF}} \|\Delta_F(\rho)\|_1 &= \max_{\stackrel{\sigma \in \bb{D}_{FE}}{j \in \bb{Z}_D}}  \| \mc{I}_F\otimes\mc{T}_j (\sigma) - \sigma \|_1 
\end{align}
As \cref{eq:NonUniformDiamond1} holds for any $F$, we have~\cite{Graydon2022}
\begin{align}
    \|\Theta(\mc{M})-\mc{M}\|_\diamond &= \max_{j \in \bb{Z}_D}  \| \mc{T}_j - \mc{I}_E \|_\diamond \notag\\
    &= \max_j 1 - \mc{F}(J(\mc{T}_j), J(\mc{I}_E)).
\end{align}
In contrast, by \cref{cor:nonuniform} we have
\begin{align}
    \sqrt{\mc{F}(J(\Theta(\mc{M})), J(\mc{M}))} &= \frac{1}{D} \sum_{j \in \bb{Z}_D}  \sqrt{\mc{F}(J(\mc{T}_j), J(\mc{I}_E))},
\end{align}
so that \cref{thm:stochasticDiamondDistance} is violated as claimed.

\section{Conclusion}\label{sec:conclusion}

In this paper, we have generalized the operational interpretation of the diamond distance and process fidelity as a probability of an error from stochastic channels to uniform stochastic instruments.
The restriction to uniform instruments plays an important role, because otherwise, the probability of an error depends upon the outcome that is observed.
As general measurement errors, including errors on idling systems, can be tailored into an effective error that is a uniform stochastic channel using a recent generalization of randomized compiling~\cite{beale2023randomized}, our result thus gives a way of meaningfully evaluating the error rates of measurements in realistic quantum computers.

\bibliography{library.bib}

\appendix

\section{Improved Fuchs-van de Graaf inequality}
\label{sec:Fuchs}
Quantum states are frequently compared by either the trace distance or the fidelity.
These two quantities are related to each other via the Fuchs-van de Graaf inequality~\cite{Fuchs1999}.
The Fuchs-van de Graaf inequality has a tighter form when one of the states is a pure state, which is in part due to the fact that the fidelity between a pure and a mixed state is just the trace inner product between them.
A natural approach to proving \cref{thm:stochasticDiamondDistance}, following proofs for unitary channels, would then be to apply the Fuchs-van de Graaf inequality to the Choi states of the ideal and noisy quantum instruments.
As the Choi state of a unitary channel is a pure state by \cref{eq:ChoiState} (since it only has one Kraus operator), a comparison of an implementation of the unitary channel to its ideal can thus make use of this improved bound.
However, the Choi state of a quantum instrument is generally not a pure state and so we cannot use the improved lower bound.
In trying this approach, we obtained the following improved lower bound on the Fuchs-van de Graaf inequality which may be of independent interest.
For clarity, we do not translate the lower bound into a bound on fidelity, but rather leave it as a trace inner product.

\begin{lemma}[Fuchs-van de Graaf Inequality]\label{thm:Fuchs}
For any two density operators $\rho, \sigma \in \bb{D}_d$ and any projector $\pi$ such that $\pi \rho = \rho$, the trace distance satisfies
\begin{align*}
    1 - \tr \pi \sigma \leq \frac{1}{2} \left\|\rho - \sigma\right\|_1\leq \sqrt{1-\|\sqrt{\rho}\sqrt{\sigma}\|_1^2}.
\end{align*}
\end{lemma}

\begin{proof}
The upper bound is simply the standard upper bound from the Fuchs-van de Gaaf inequality with the definition of the fidelity substituted in.
Recall that for positive semi-definite matrices $A$ and $B$ we have the following identity for the trace norm~\cite{MI}
\begin{align}
    \frac{1}{2}\|A - B\|_1 \geq \max_{\pi\leq I} \tr\left(\pi(A - B)\right).
\end{align}
Then choosing $\pi$ such that $\pi \rho = \rho$ gives
\begin{align}
    \frac{1}{2}\|\rho - \sigma\|_1 \geq \tr \pi(\rho - \sigma),
\end{align}
giving the lower bound.
\end{proof}

\end{document}